\documentclass[a4paper,UKenglish,cleveref, autoref, thm-restate]{lipics-v2021}
\hideLIPIcs
\nolinenumbers
\usepackage{mathtools}

\title{On the compressiveness of the Burrows-Wheeler transform}

\author{Hideo Bannai}{M\&D Data Science Center, Institute of Integrated Research, Institute of Science Tokyo, Japan}{hdbn.dsc@tmd.ac.jp}{https://orcid.org/0000-0002-6856-5185}{JSPS KAKENHI Grant Number JP24K02899}
\author{Tomohiro I}{Kyushu Institute of Technology, Japan}{tomohiro@ai.kyutech.ac.jp}{https://orcid.org/0000-0001-9106-6192}{}
\author{Yuto Nakashima}{Department of Informatics, Kyushu University, Japan}{nakashima.yuto.003@m.kyushu-u.ac.jp}{https://orcid.org/0000-0001-6269-9353}{JSPS KAKENHI Grant Number JP21K17705 and JP23H04386}

\authorrunning{H. Bannai, T. I, Y. Nakashima}


\category{} 

\relatedversion{}

\EventEditors{John Q. Open and Joan R. Access}
\EventNoEds{2}
\EventLongTitle{42nd Conference on Very Important Topics (CVIT 2016)}
\EventShortTitle{CVIT 2016}
\EventAcronym{CVIT}
\EventYear{2016}
\EventDate{December 24--27, 2016}
\EventLocation{Little Whinging, United Kingdom}
\EventLogo{}
\SeriesVolume{42}
\ArticleNo{23}

\newcommand{\rank}{\mathsf{rank}}
\newcommand{\rot}{\mathsf{rot}}
\newcommand{\bwt}{\mathsf{BWT}}
\newcommand{\bbwt}{\mathsf{BBWT}}
\newcommand{\source}{\mathit{s}}
\newcommand{\pos}{\mathsf{pos}}
\begin{document}

\maketitle

\begin{abstract}
    The Burrows-Wheeler transform (BWT) is a reversible transform that converts a string $w$ into another string $\bwt(w)$.
    The size of the run-length encoded BWT (RLBWT) can be interpreted as a measure of repetitiveness in the class of representations called dictionary compression
    which are essentially representations based on copy and paste operations.
    In this paper, we shed new light on the compressiveness of BWT
    and the bijective BWT (BBWT).
    We first extend previous results on the relations of their run-length compressed sizes $r$ and $r_B$.
    We also show that the so-called ``clustering effect'' of BWT and BBWT can be
    captured by measures other than empirical entropy or run-length encoding.
    In particular, we show that
    BWT and BBWT do not increase the repetitiveness of the string with respect to
    various measures based on dictionary compression by more than a polylogarithmic factor.
    Furthermore, we show that there exists an infinite family of strings
    that are maximally incompressible by any dictionary compression measure,
    but become very compressible after applying BBWT.
    An interesting implication of this result is that it is possible
    to transcend dictionary compression in some cases by simply
    applying BBWT before applying dictionary compression.
\end{abstract}

\setcounter{page}{0}
\clearpage
\section{Introduction}
The Burrows-Wheeler Transform (BWT)~\cite{bwt94} is a reversible mapping from a string to another string
that enables compression and efficient pattern search,
and is the theoretical cornerstone for essential tools in the field of bioinformatics~\cite{bowtie09,10.1093/bioinformatics/btp324}.
The compressibility of BWT has been studied in various contexts,
but more recently,
rather than statistical measures such as empirical entropy which are not helpful in highly repetitive datasets,
the size $r$ of the run-length encoded BWT (RLBWT) and its relation to the many other
repetitiveness measures related to dictionary compression has become an important topic of study
(See~\cite{DBLP:journals/csur/Navarro21} for a comprehensive survey).

Dictionary compression is a family of compressed representations which are essentially based on copy and paste operations.
Kempa and Prezza~\cite{DBLP:conf/stoc/KempaP18} proposed the notion of
    {\em string attractors} to
view dictionary compression in a uniform way,
and showed that the size $\gamma$ of the smallest string attractor
lower bounds all other measures of dictionary compression, since they implicitly give a
string attractor of the same size as their representations:
the size $b$ of the smallest bidirectional macro scheme (BMS)~\cite{DBLP:journals/jacm/StorerS82},
the size $z$ of the LZ77 parsing~\cite{DBLP:journals/tit/ZivL77},
the size $g_{\mathrm{rl}}$ of the smallest grammar with run-length rules,
the size $g$ of the smallest grammar,
$r$, and more.
Another measure $\delta$ based on substring complexity is known to lowerbound $\gamma$~\cite{9961143}.

Concerning $r$ in relation to the other measures,
Navarro et al.~\cite{DBLP:journals/tit/NavarroOP21} showed that
given an RLBWT of size $r$, we can construct a BMS of size $2r$,
thus showing $b = O(r)$.
Kempa and Kociumaka~\cite{DBLP:journals/cacm/KempaK22} showed $r = O(z\log^2 n)$,
and further $r=O(\delta\log^2 n)$\footnote{More precisely, $r=O(\delta\log \delta\max(1,\log\frac{n}{\delta\log \delta}))$}.
Also, $z = O(b\log n)$~\cite{DBLP:journals/tit/NavarroOP21} has been shown,
and therefore $z = O(r \log n)$ or, $r  = \Omega(z/\log n)$.
Recently, the size $r_B$ of the run-length encoded bijective BWT (RLBBWT) has been studied~\cite{DBLP:conf/ictcs/0002CLR23,DBLP:conf/spire/BadkobehBK24},
and was shown, similarly to $r$, that $r_B = O(z\log^2 n)$~\cite{DBLP:conf/spire/BadkobehBK24}. Also, the existence of an infinite family of strings such that $r = o(r_B)$, namely, $r = O(1)$ and $r_B = \Omega(\log n)$ was shown~\cite{DBLP:conf/spire/BadkobehBK24}. The existence of the opposite case,
i.e., if there are string families for
which $r_B = o(r)$, was left open.

The ``clustering effect'' of the BWT in terms of the run length encoding was studied by Mantaci et al.~\cite{DBLP:journals/tcs/MantaciRRSV17},
where they
proved that BWT can increase the size of the run-length encoding of the string by a factor of at most $2$.
The maximal clustering effect for run-length encoding can be seen for example with the Fibonacci words: Fibonacci words of length $n$
have a run-length encoding of size $\Theta(n)$,
while $r=2$.
While this is remarkable when viewed from run-length encoding
since a maximally incompressible string (in terms of rle) is
transformed into a compressible one,
in terms of the smallest bidirectional macro scheme,
Fibonacci words and their BWT have the same size $4$,
and one might argue that BWT is merely transforming a compressible string into another compressible string.

In this paper, we extend
previous results on the relation between $r$ and $r_B$.
We also investigate the clustering effect of BWT and BBWT in terms of other repetitiveness measures
and show that they can be much more powerful than previously perceived,
giving an example where
BBWT transforms a maximally incompressible string
(in terms of dictionary compression)
into a compressible one.

The contributions of this paper are as follows.
\begin{enumerate}
    \item
          We show an infinite family of strings such that
          $r = \Omega(\log n)$ and $r_B = O(1)$,
          answering affirmatively an open question about the existence of
          a family of strings where $r_B = o(r)$ raised by Badkobeh et al.~\cite{DBLP:conf/spire/BadkobehBK24}.

    \item We show a polylogarithmic ($O(\log^4 n)$) upper bound on the multiplicative difference between $r$ and $r_B$.

    \item We analyze the clustering effect of BWT and BBWT in terms of other repetitiveness measures,
          and show that the application of BWT and BBWT can only increase various repetitiveness measures of a string by a polylogarithmic factor.
    \item
          We show that BBWT on the strings whose BBWT
          images are Fibonacci words
          exhibit a maximal clustering effect in terms of the repetitiveness measures
          $\delta$, $\gamma$, $b$, $v$, or $r$,
          from almost linear ($\Theta(n/\log n)$) to constant.          The main combinatorial part of our proof is an elegant characterization of the LF mapping function on Fibonacci strings that uses the Zeckendorf representation of the positions.

          As a byproduct, this result shows that
          it is possible in some cases
          to transcend dictionary compression
          by simply applying BBWT before applying dictionary compression.
\end{enumerate}

\section{Preliminaries}
Let $\Sigma$ denote a set of symbols called the alphabet.
A string is an element of $\Sigma^*$.
If $w = xyz$ for any strings $w,x,y,z$, then,
$x$, $y$, $z$ are respectively a
    {\em prefix}, {\em substring}, and {\em suffix} of $w$ and are a {\em proper} prefix, substring, or suffix
if they are not equal to $w$.
The length of string $x$ is denoted by $|x|$.
The empty string, i.e., the string of length $0$ is denoted by $\varepsilon$.
The symbol at position $i$ will be denoted by $x[i]$, and we will use a $0$-based index, i.e.,
$x = x[0]\cdots x[|x|-1]$.
For integers $i,j\in [0,|x|)$, let $x[i..j] = x[i]\cdots x[j]$ if $i \leq j$
and $x[i..j] = \varepsilon$ if $i > j$. We will also use $x[i..j) = x[i..j-1]$.
For any string $x$, $x^0 = \varepsilon$,
and for any integer $i \geq 1$, $x^i = x^{i-1}x$.
A string $w$ is {\em primitive} if it cannot be represented as $w=x^k$ for some string $x$ and integer $k\geq 2$.
For any symbol $c\in\Sigma$ let $|x|_c = |\{ i \mid x[i] = c\}|$,
i.e., the number of $c$'s in $x$.

For any string $x$, let $\rot(x) = x[1..|x|)x[0]$,
denote a left cyclic rotation of $x$ by one symbol.
Notice that for any integer $i$, $\rot^i(x)$ naturally corresponds to the cyclic
rotation of $x$ by $i$ symbols to the left if $i > 0$, and to the right if $i < 0$.
We will use $\rot^*(x)$ to denote the lexicographic smallest rotation of $x$.
A {\em conjugacy class} is an equivalence class of the equivalence relation defined by
$x\equiv y \iff \rot^*(x) = \rot^*(y)$.

A string $w$ is a {\em Lyndon word} if it is lexicographically smaller than any of its,
proper rotations, i.e., $w < \rot^k(w)$ for any $k\in [1,|w|)$.
From the definition, it is clear that a Lyndon word must be primitive.
A rotation of a primitive string $x$ is always primitive and $\rot^*(x)$ is Lyndon and unique, which we will
sometimes refer to as the Lyndon rotation of $x$.
It is known that Lyndon words cannot have a non-empty proper border
(a proper prefix that is also a proper suffix).
The Lyndon factorization~\cite{chen58lyndon}
of a string $w$ is a unique partitioning of $w$
into a sequence of non-increasing Lyndon words,
i.e., $w = L^{e_1}_1 \cdots L^{e_{\ell(w)}}_{\ell(w)}$
where each $L_i~(i\in [1,\ell(w)])$,
which we will call {\em Lyndon factors} of $w$,
is a Lyndon word and $L_{i} > L_{i+1}$ for $i \in [1,\ell(w))$.
It is known that any occurrence of a Lyndon word
in $w$ must be a substring of a Lyndon factor of $w$.

The $\omega$-order $<_\omega$
between primitive strings
or same length strings
$x,y$ is defined as $x<_\omega y \iff x^\infty < y^\infty$.\footnote{As we will not be comparing non-primitive strings of different lengths in this paper, the definition here is a simplified version of the original.}
The $\omega$-order can be different from the standard lexicographic order,
but are identical when comparing strings of the same length or when comparing Lyndon words.

For a string $x[0..|x|)$, let $\rank_{c}(i,x) = |x[0..i)|_c$,
i.e., the number of symbols $c$ in $x[0..i)$.
Let $\rho(x)$ denote the size of the run-length encoding of $x$, i.e.,
the maximal number of same symbol runs in $x$.
Let $\rho_c(x)$ denote the number of runs of symbol $c$ in the run-length encoding of $x$.

\subsection{Repetitiveness Measures}
A set of positions $\Gamma$ is a string attractor~\cite{DBLP:conf/stoc/KempaP18} of $w$ if any substring of $w$ has an occurrence in $w$ that covers a position in $\Gamma$.
The size of the smallest string attractor of $w$ is denoted by $\gamma(w)$.
The measure $\delta$~\cite{9961143} is defined as $\max_{k\in[1,|w|]} S_k/k$, where $S_k$ is the number of distinct length-$k$ substrings of $w$.

The Burrows-Wheeler transform (BWT)~\cite{bwt94} $\bwt(w)$ of a string $w$ is defined as the
sequence of last (or equivalently, previous, in the cyclic sense) symbols
of all rotations of $w$, in lexicographic order of the rotations.
The size of the run-length encoding of $\bwt(w)$ will be denoted by $r(w)$,
i.e., $r(w) = \rho(\bwt(w))$.
The bijective BWT (BBWT)~\cite{DBLP:journals/corr/abs-1201-3077}
$\bbwt(w)$ of a string $w$ is defined as the
sequence of last (or again, previous, in the cyclic sense) symbols of all the rotations of all the Lyndon factors of $w$, in $\omega$-order of the rotations.
Actually, BWT can be understood as a special case of
BBWT: $\bwt(w) = \bbwt(\rot^*(w))$
because $\rot^*(w) = L^e$ for some Lyndon word $L$ and integer $e = |w|/|L|$,
and the $\omega$-order is equivalent to lexicographic order in this case
since all the compared strings are of the same length.
The size of the run-length encoding of $\bbwt(w)$ will be denoted by $r_B(w)$,
i.e., $r_B(w) = \rho(\bbwt(w))$.

The inverse transform of $\bwt$ and $\bbwt$ on a string $x$ can be defined by the {\em LF mapping},
which is a function $\Psi_x(i) = j$
over positions of $x$
where $c = x[i] = s[j]$, $\rank_{c}(i,x) = \rank_{c}(j,s)$,
and $s$ is the string obtained by sorting the multiset of symbols of $x$ in increasing order.
$\Psi_x$ is a permutation and thus forms cycles on the set of positions of $x$.
A cycle $(i, \Psi^1(i), \cdots, \Psi^{k-1}(i))$
where $k$ is the smallest positive integer such that $\Psi^{k}(i) = i$,
corresponds to a (cyclic) string $x[\Psi^{k-1}(i)]\cdots x[\Psi^1(i)] x[i]$.
This string is always primitive, and by concatenating,
in non-increasing order, all the Lyndon rotations of the strings
corresponding to all cycles, it can be shown that $\bbwt^{-1}(x)$ is obtained.
$\Psi_x$ can be interpreted as returning, given the $\omega$-order rank
of a given rotation of a cycle, the $\omega$-order rank of the previous rotation of the cycle.
Note that when $x$ is a $\bwt$ image of a primitive string, $\Psi_x$ consists of only a single cycle and that although the Lyndon rotation will give the string $w$ for which $\bbwt(w) = x$,
we have $\bwt(w') = x$ for any rotation $w'$ of $w$.

A Bidirectional Macro Scheme (BMS)~\cite{DBLP:journals/jacm/StorerS82} of a string $w$
is a partitioning of $w$ into phrases,
where each phrase is either a single symbol, or is a substring that
has an occurrence elsewhere in $w$ which we call the source, or the reference of the phrase.
The references of the phrases must
be such that the induced reference for each position in the phrases
are acyclic, i.e., the referencing on the position forms a
forest where the roots are the positions corresponding to single symbol phrases.
The size of the smallest bidirectional macro scheme for $w$ is denoted by $b(w)$.

The LZ77 factorization~\cite{DBLP:journals/tit/ZivL77} of a string $w$ is a BMS
of $w$ where all references are left-referencing, i.e.,
they must point to a smaller position,
and the phrases are determined greedily from left-to-right.
It is known that LZ77 is the smallest among left-referencing BMS.
The size of LZ77 of $w$ is denoted by $z(w)$.

The lex-parse~\cite{DBLP:journals/tit/NavarroOP21}
of a string $w$ is a BMS of $w$ where all references point to a rotation
of smaller lexicographic (or equivalently $\omega$-order) rotation,
and the phrases are determined greedily from left-to-right.
Similarly, it is known that lex-parse is the smallest among BMS with such constraint.
The size of lex-parse of $w$ is denoted by $v(w)$.

\subsection{Fibonacci Words}

Since we will later use Fibonacci words or their slight modifications
to show some of our results, we introduce them here.

The Fibonacci words are defined recursively as:
$F_0 = \mathtt{b}$, $F_1 = \mathtt{a}$, and for any integer $i \geq 2$, $F_i = F_{i-1} F_{i-2}$.
Fibonacci words can also be defined via a morphism $\phi$ defined as
$\phi(\mathtt{a}) = \mathtt{ab}$, and $\phi(\mathtt{b}) = \mathtt{a}$,
and $F_i = \phi^i(\mathtt{b})$ for all $i \geq 0$.
The lengths of the Fibonacci words correspond to the Fibonacci sequence,
i.e., $f_0 = f_1 = 1$, and $f_i = f_{i-1}+f_{i-2}$.
For technical reasons, we define $f_i = 0$ for $i < 0$.
The observation below follows from a simple induction.
\begin{observation}\label{observation:NumberOfAorB}
    For any $i \geq 1$,
    $|F_i|_\mathtt{a} = f_{i-1}$, and
    $|F_i|_\mathtt{b} = f_{i-2}$.
\end{observation}

\section{New results for $r$ vs $r_B$}
\subsection{Strings family giving $r(w)/r_B(w) =\Omega(\log n)$}
Here, we answer an open question raised by Badkobeh et al.~\cite{DBLP:conf/spire/BadkobehBK24}, and show that
there exists an infinite family of strings such that
$r_B$ is asymptotically strictly smaller than $r$.

\begin{restatable}{theorem}{rBSmallerthanR}\label{theorem:rB_smallerthan_r}
    There exists an infinite family of strings such that $r = \Omega(r_B\log n)$.
\end{restatable}
\begin{proof}
    Consider string
    $w_k = \mathtt{b}F^*_{2k}$, where $F^*_{2k} = \rot^*(F_{2k})$.
    Then, for any $k\geq 3$,
    $r_B(w_k) = 3$ (\Cref{lemma:rBConstant}),
    and $r(w_k) = 2k$ (\Cref{corollary:rLogarithmic}).
\end{proof}

\begin{lemma}\label{lemma:rBConstant}
    Let $w_k = \mathtt{b} F^*_{2k}$, where $F^*_{2k} = \rot^*(F_{2k})$.
    Then, $r_B(w_k) = 3$ for any $k \geq 3$.
\end{lemma}
\begin{proof}
    The Lyndon factorization of $w_k$ results in
    the factors $\mathtt{b}$ and $F^*_{2k}$,
    and it is known that $\bbwt(F^*_{2k}) = \bwt(F^*_{2k}) = \mathtt{b}^{f_{2k-2}}\mathtt{a}^{f_{2k-1}}$.
    Since $\mathtt{b}$ is greater, in $\omega$-order, than any rotation of $F^*_{2k}$, we have
    $\bbwt(w_k) = \mathtt{b}^{f_{2k-2}}\mathtt{a}^{f_{2k-1}}\mathtt{b}$ and $r_B(w_k) = 3$.
\end{proof}
The rest of the proof will focus on showing $r(w_k) = 2k$.
We show two proofs, one that relies heavily on previous results,
and the other an alternate direct proof.

\subsubsection{Proof for $r(w_k) = 2k$}
Actually, it turns out that $w_k$ is a rotation of the strings whose BWT
are shown to have $2k$ runs
in Proposition 3 of~\cite{DBLP:conf/sofsem/GiulianiILPST21} and
in Proposition 3
of~\cite{DBLP:conf/dlt/GiulianiILRSU23}.
The former defines it as $v_{2k}\mathtt{b}$ for $v_{2k} = F_{2k}^R$.
The latter defines it as inserting $\mathtt{b}$ at position $f_{2k-1}-2$ of $F_{2k}$,
and mentions that it is a rotation of the former string.
Our proof of $r(w_k) = 2k$ further connects these strings with the lexicographically smallest rotation
$F^*_{2k}$ of $F_{2k}$.

The following is known:
\begin{theorem}[Theorem~1 in~\cite{christodoulakis_sofsem2006}]
    \label{theorem:fibonacciRank}
    The rotation of $F_n$ with rank $\rho$ in the lexicographically sorted list of all the rotations of $F_n$, for $n\geq 2$, $\rho \in [0..f_n)$ is the rotation $\rot^i(F_n)$ where
    \[
        i = \begin{cases}
            (\rho\cdot f_{n-2}-1)\bmod f_n      & \mbox{ if $n$ odd}   \\
            (-(\rho+1)\cdot f_{n-2}-1)\bmod f_n & \mbox{ if $n$ even}.
        \end{cases}
    \]
\end{theorem}
Therefore, we have
\begin{corollary} \label{coro:minimum-rot}
    For every integer $k \geq 1$,
    $F^*_{2k} = \rot^{f_{2k-1}-1}(F_{2k})$.
\end{corollary}
\begin{proof}
    From~\Cref{theorem:fibonacciRank} with $\rho = 0$, $F^*_{2k} = \rot^i(F_{2k})$ where
    $i = -(f_{2k-2}+1) \equiv f_{2k-1}-1 \pmod {f_{2k}}$.
\end{proof}
which means that inserting a $\mathtt{b}$ at position $f_{2k-1}-2$ of $F_{2k}$ is a rotation of $\mathtt{b}F^*_{2k}$,
leading to:
\begin{corollary}\label{corollary:rLogarithmic}
    Let $w_k = \mathtt{b}\cdot F^*_{2k}$, where $F^*_{2k} = \rot^*(F_{2k})$.
    Then, $r(w_k) = 2k$ for any $k \geq 3$.
\end{corollary}

\subsubsection{Alternate proof for $r(w_k)=2k$}
We also give an alternate direct proof based on morphisms, which is perhaps slightly simpler compared to the proof for $r(v_{2k}\mathtt{b})=2k$ presented in~\cite{DBLP:conf/sofsem/GiulianiILPST21}
which relies on additional results on {\em special factors} of standard words~\cite{DBLP:journals/ita/BorelR06}.
Results on morphisms and their effects on $r$ have been studied by Fici et al.~\cite{DBLP:conf/cpm/FiciRSU23}, but to the best of our knowledge,
their results do not directly apply to our case.

We use a string morphism $\theta$ defined as:
$\theta(\mathtt{a}) = \mathtt{aab}$, $\theta(\mathtt{b}) = \mathtt{ab}$.
The following claim can be shown by a simple induction.

\begin{claim} [Claim~5 in \cite{mieno_cpm2022}]
    \label{coro:even-morphism}
    For any string $w \in \{\mathtt{a}, \mathtt{b}\}^*$, $\theta(w) = \rot^2(\phi^2(w))$.
\end{claim}

To prove \Cref{theorem:rB_smallerthan_r}, we first show the following lemma.
\begin{lemma}\label{lem:min-rot_morphism}
    For every integer $k \geq 0$,
    $F^*_{2k} = \theta^k(\mathtt{b})$.
\end{lemma}
\begin{proof}
    It clearly holds for $k=0,1$.
    For $k \geq 2$, we show that $\theta^k(\mathtt{b}) = \rot^{f_{2k-1}-1}(F_{2k})$ by induction on $k$.
    We can see that the statement holds for $k=2$ since $\theta^2(\mathtt{b}) = \mathtt{aabab} = \rot^2(F_4)$.
    Suppose that $\theta^{k'}(\mathtt{b}) = \rot^{f_{2k'-1}-1}(F_{2k'})$ holds for some $k' \geq 2$.
    We have
    \begin{align*}
        \theta^{k'+1}(\mathtt{b})
         & = \theta(\theta^{k'}(\mathtt{b})) = \theta(\rot^{f_{2k'-1}-1}(F_{2k'})) &  & \text{by induction hypothesis}      \\
         & = \theta(\rot^{-1}(F_{2k'-2}F_{2k'-1}))                                                                          \\
         & = \rot^2(\phi^2(\rot^{-1}(F_{2k'-2}F_{2k'-1}))                          &  & \text{by~\Cref{coro:even-morphism}}
\end{align*}
    Since the last symbol of $F_{2k'}$ is $F_{2k'}[f_{2k'-1}-1] = \mathtt{a}$ and $|\phi^2(\mathtt{a})| = 3$,
    we have
    \begin{align*}
        \theta^{k'+1}(\mathtt{b})
         & =\rot^2(\rot^{-3}(F_{2k'}F_{2k'+1}))  \\
         & =\rot^{f_{2k'+1}-1}(F_{2k'+2})        \\
         & = \rot^{f_{2(k'+1)-1}-1}(F_{2(k'+1)})
\end{align*}
    Hence, the statement also holds for $k = k'+1$,
    and $\theta^k(\mathtt{b}) = \rot^{f_{2k-1}-1}(F_{2k})$ for every $k \geq 2$.
    By combining it with Corollary~\ref{coro:minimum-rot},
    we obtain $F^*_{2k} = \theta^k(\mathtt{b})$.
\end{proof}
We are ready to prove the following lemma.
\begin{lemma}\label{lemma:rLogarithmic}
    Let $w_k = \mathtt{b}F^*_{2k}$, where $F^*_{2k} = \rot^*(F_{2k})$.
    Then, $r(w_k) = 2k$ for any $k \geq 2$.
\end{lemma}
\begin{proof}
    We first prove that $r(y_k) = 2k+2$ for $y_k = \mathtt{c}F^*_{2k}$
    for any $k \geq 3$.
    From Lemma~\ref{lem:min-rot_morphism}, we have $y_k = \mathtt{c}F^*_{2k} = \mathtt{c}\theta^k(\mathtt{b})$.
    Below, we extend the definition of $\theta$ so that $\theta(\mathtt{c})=\mathtt{c}$,
    so $y_k = \theta^k(\mathtt{cb})$.
    We claim that
    \begin{align}
        \bwt(y_k) & = \mathtt{c}\mathtt{b}^{f_{2k -2} - k} \prod_{j=1}^{k} (\mathtt{a}^{f_{2j-2}}\mathtt{b})\label{eqn:bwtyk}
    \end{align}
    We show this statement by induction on $k$.
    If $k=3$, $y_3 = \mathtt{caabaababaabab}$ and $\bwt(y_3) = \mathtt{cb}^2\mathtt{aba}^2\mathtt{ba}^5\mathtt{b}$ hold.
    Assume that the statement holds for some $k' \geq 3$.
    We show the statement also holds for $k=k'+1$.
    We consider the three disjoint sets of the rotations of $y_{k'+1}$ as follows: the rotations (i) starting with $\mathtt{a}^2$, (ii) starting with $\mathtt{ab}$, (iii) starting with $\mathtt{b}$ or $\mathtt{c}$.
    \begin{enumerate}[(i)]
        \item
              From the definition of $\theta$,
              the number of rotations starting with $\mathtt{a}^2$ of $y_{k'+1}$ is $f_{2k'-1}~(= |F_{2k'}|_\mathtt{a})$.
              These rotations are the first $f_{2k'-1}$ rotations in the lexicographically increasingly sorted list of all rotations of $y_{k'+1}$.
              By the definition of $y_{k}$, the lexicographically smallest rotation is $\rot(y_{k'})$,
              and the preceding symbol of this rotation is $\mathtt{c}$.
              The other rotations of this case are preceded by $\mathtt{b}$.
              Hence, $\bwt(y_{k'+1})[0..f_{2k'-1}-1] = \mathtt{cb}^{f_{2k'-1}-1}$.

        \item Let $\pos$ be the set of positions in $y_{k'+1}$ that are preceded
              by an occurrence of $\mathtt{ab}$.
              From the definitions of $y_{k}$ and $\theta$,
              every occurrence of $\mathtt{ab}$ in $y_{k'+1}$
              is a suffix of an occurrence of the substrings ($\mathtt{aab}$ or $\mathtt{ab}$) produced from $\theta$ and an occurrence of $\mathtt{a}$ or~$\mathtt{b}$ in~$y_{k'}$.
Therefore, we have $|\pos| = |y_{k'}|-1 = f_{2k'}$ and
              if $\source(i)$ denotes the position in $y_{k'}$ that produces the corresponding symbol at position $i$ in $y_{k'+1}$ by the morphism $\theta$, then
$\rot^{i}(y_{k'+1}) = \theta(\rot^{\source(i)}(y_{k'}))$ for any $i\in\pos$.
              Thus, for any positions $i_1, i_2 \in \pos$,
              \begin{align}
                  \rot^{i_1-2}(y_{k'+1}) < \rot^{i_2-2}(y_{k'+1})
                   & \iff \rot^{i_1}(y_{k'+1}) < \rot^{i_2}(y_{k'+1})\nonumber                                \\
                   & \iff \theta(\rot^{\source(i_1)}(y_{k'})) < \theta(\rot^{\source(i_2)}(y_{k'}))\nonumber  \\
                   & \iff \rot^{\source(i_1)}(y_{k'}) < \rot^{\source(i_2)}(y_{k'}).\label{eqn:equaltosource}
              \end{align}
              where the last relation follows from the fact that $\theta$ is an order preserving morphism (i.e., $s < t \Leftrightarrow \theta(s) < \theta(t)$ holds for any $s,t$).
We can also see that for any $i\in\pos$,
              the symbol $y_{k'+1}[i-3]$ preceding the occurrence of $\mathtt{ab}$ is
              $\mathtt{a}$ (resp. $\mathtt{b}$) iff $y_{k'}[\source(i)-1]$ is $\mathtt{a}$ (resp. $\mathtt{b}$).
              Thus, together with \Cref{eqn:equaltosource},
              the sequence of entries in $\bwt(y_{k'+1})$ corresponding to rotations that start with $\mathtt{ab}$ are equivalent to
              the sequence of entries in $\bwt(y_{k'})$ for rotations that start with $\mathtt{a}$ or $\mathtt{b}$,
              i.e.,
              $\bwt(y_{k'+1})[f_{2k'-1}..f_{2k'+1}-1] = \bwt(y_{k'})[1..f_{2k'}] = \mathtt{b}^{f_{2k'-2} - k'} \prod_{j=1}^{k'} (\mathtt{a}^{f_{2j-2}}\mathtt{b})$.
              See also~\Cref{fig:large-bwt}.
        \item
              From the definition of $\theta$,
              the number of rotations starting with $\mathtt{b}$ in $y_{k'+1}$ is $f_{2k'}$,
              and the lexicographically largest rotation is $y_{k'+1}$ itself.
              These rotations are the last $f_{2k'}+1$ rotations in the lexicographically increasingly sorted list of all rotations of $y_{k'+1}$.
              The preceding symbol of the largest rotation, or equivalently, the last symbol of $y_k$ is $\mathtt{b}$.
              All rotations starting with $\mathtt{b}$ are preceded by $\mathtt{a}$.
              Hence, $\bwt(y_{k'+1})[f_{2k'+1}..f_{2k'+2}] = \mathtt{a}^{f_{2k'}}\mathtt{b}$.
    \end{enumerate}
    All together, we have
    \begin{align*}
        \bwt(y_{k'+1})
         & = \overbrace{\mathtt{cb}^{f_{2k'-1}-1}\vphantom{\prod_{j}^{k'}]}}^{\mathsf{(i)}}
        \cdot
        \overbrace{\mathtt{b}^{f_{2k'-2}-k'} \prod_{j=1}^{k'} (\mathtt{a}^{f_{2j-2}}\mathtt{b})}^{\mathsf{(ii)}}
        \cdot       \overbrace{\mathtt{a}^{f_{2k'}}\mathtt{b}\vphantom{\prod_{j}^{k'}]}}^{\mathsf{(iii)}} \\
         & = \mathtt{c}\mathtt{b}^{f_{2k'} - (k'+1)} \prod_{j=1}^{k'+1} (\mathtt{a}^{f_{2j-2}}\mathtt{b})
    \end{align*}
    and the statement holds for $k=k'+1$, proving \Cref{eqn:bwtyk}.
    Therefore, $r(y_k) = 2k+2$.
    We note that \Cref{eqn:bwtyk} holds for $k=2$, but since $f_{2k-2}-2 = 0$,
    $r(y_2) = 5$.

    The last symbol of $F^*_{2k}$ is $\mathtt{b}$, so changing the $\mathtt{c}$ in $y_k$ to $\mathtt{b}$
    changes the unique occurrence of $\mathtt{bc}$ in all rotations of $y_k$ except for $y_k$ itself,
    to the unique occurrence of $\mathtt{bb}$ in all rotations of $w_k$ except for $w_k$ itself.
    Therefore, the lexicographic order of rotations of $w_k$ are equivalent to those in $y_k$,
    except for $w_k$ itself.
    We can see that $\bwt(w_k)$ will have the following changes compared to $\bwt(y_k)$:
    1) the lexicographically smallest rotation $\rot(w_k)$ is preceded by $\mathtt{b}$ instead of $\mathtt{c}$,
    and
    2) $w_k$ is the lexicographically smallest rotation that starts with $\mathtt{b}$ (rather than $\mathtt{c}$).
1) changes $\mathtt{c}$ to $\mathtt{b}$ and thus decreases the run-length by $1$.
    2) moves the last $\mathtt{b}$ right after the second to last $\mathtt{b}$ in $\bwt(y_k)$
    (the preceding symbol of the largest rotation starting with $\mathtt{a}$) and thus decreases the run-length by $1$.
    Therefore, in total, $r(w_k) = 2k$ which can be confirmed to hold for $k=2$ as well.
\end{proof}

\begin{figure}[t]
    \begin{center}
        \includegraphics[width=0.8\textwidth]{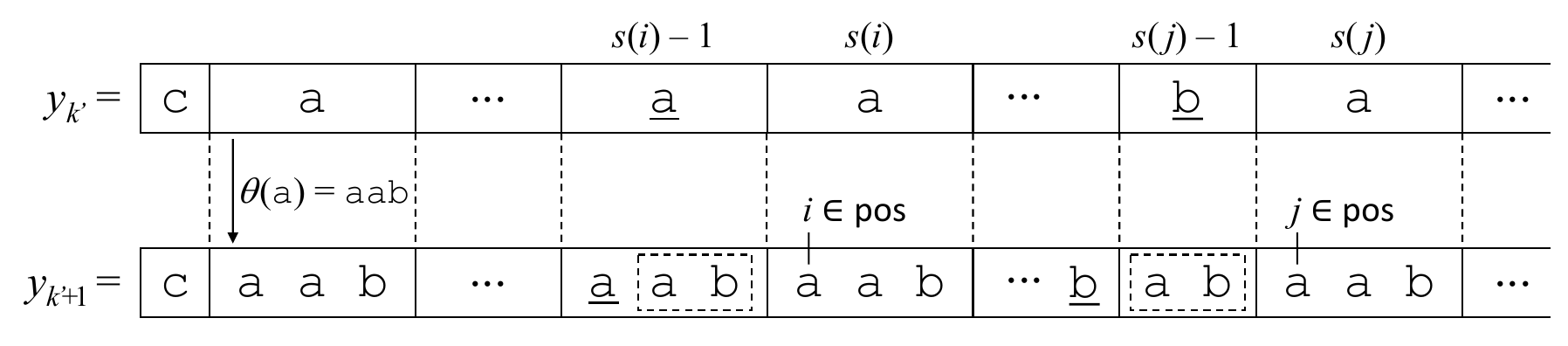}
        \caption{
            Illustration of the proof of Lemma~\ref{lemma:rLogarithmic}.
            Every position in the set $\pos$ is preceded by an occurrence of $\mathtt{ab}$.
            We can obtain the lexicographic rank of the rotation $\rot^{i-2}(y_{k'+1})$ that starts with $\mathtt{ab}$ by using the corresponding rotation $\rot^{s(i)}(y_{k'})$.
            The preceded symbols (underlined symbols) of $y_{k'+1}[i-3]$ and $y_{k'}[\source(i)-1]$ are the always same.
            The figure shows the case of $y_{k'+1}[j-3] = y_{k'}[\source(j)-1] = \mathtt{a}$ by the position $i$ and the case of $y_{k'+1}[j-3] = y_{k'}[\source(j)-1] = \mathtt{b}$ by the position $j$.
        }
        \label{fig:large-bwt}
    \end{center}
\end{figure}

\subsection{Bounds for $r(w)/r_B(w)$}
We show poly-logarithmic upper and lower bounds on the ratio $r(w)/r_B(w)$ for any string.

We first show a simple upper bound on the multiplicative rotation sensitivity of $z$.
\begin{lemma}\label{lemma:sensitivity:zrotation}
    For any strings $w,w'$ of length $n$
    in the same conjugacy class,
    $z(w') = O(z(w)\log n)$.
\end{lemma}
\begin{proof}
    It is easy to see that $b(w') = \Theta(b(w))$.
    This is because, given any BMS of $w$ of size $b$,
    we can reuse the parsing and referencing,
    except for the following changes,
    to construct a BMS for $w'$ of size $2b+1$:
    (1) if $w'$ starts in the middle of a BMS phrase of $w$, we split the phrase,
    and (2) if a phrase references a substring of $w$ but is split in $w'$, we split the phrase.
    Since only one of the phrases of the split phrase in (1) is split by (2),
    the total number of phrases is at most $2b+1$.

    Since
    $b(w) \leq z(w) = O(b(w)\log n)$ and
    $b(w') \leq z(w') = O(b(w')\log n)$, we have $z(w')/z(w) = O(\log n)$.
\end{proof}

\begin{lemma}\label{lemma:sensitivity:rBrotation}
    For any strings $w,w'$ of length $n$
    in the same conjugacy class,
    $r_B(w') = O(r_B(w)\log^4 n)$.
\end{lemma}
\begin{proof}
    Badkobeh et al.~\cite{DBLP:conf/spire/BadkobehBK24} showed (1) $b(w) = O(r_B(w))$, which implies $z(w)=O(r_B(w)\log n)$,
    and (2) $r_B(w) = O(z(w)\log^2 n)$.
    Therefore, $r_B(w')/r_B(w) = O(z(w')/z(w)\log^3 n) = O(\log^4 n)$
    from~\Cref{lemma:sensitivity:zrotation}.
\end{proof}
\begin{corollary}\label{corollary:RvsRB}
    For any string $w$ of length $n$, $r_B(w) = O(r(w) \log^4 n)$ and $r(w) = O(r_B(w)\log^4 n)$.
\end{corollary}

\section{Clustering effects of $\bwt$ and $\bbwt$}
\subsection{In terms of $\rho$}
The clustering effect of $\bwt$ was measured by the length $\rho$ of the run-length encoding.
Namely, the following statement was claimed by Mantaci et al.~\cite{DBLP:journals/tcs/MantaciRRSV17}.
\begin{theorem}[Theorem 3.3 in~\cite{DBLP:journals/tcs/MantaciRRSV17}]\label{thm:bwtRLClustering}
    For any string $w$, $\rho(\bwt(w)) \leq 2\rho(w)$.
\end{theorem}
While the statement is true, we believe there is
a non-trivial case that was not covered in their
original proof.
Here, we address this case and also show that the same
statement holds for $\bbwt$.
Note that the following \Cref{thm:bbwtRLClustering} implies~\Cref{thm:bwtRLClustering} because $\bwt(w) = \bbwt(\rot^*(w))$ and $\rho(\rot^*(w)) \leq \rho(w)$.
\begin{theorem}\label{thm:bbwtRLClustering}
    For any string $w$, $\rho(\bbwt(w)) \leq 2\rho(w)$.
\end{theorem}
\begin{proof}
    We first recall the arguments of~\cite{DBLP:journals/tcs/MantaciRRSV17} for BWT.
    Consider a range $[i..j]$ of $x = \bwt(w)$
    that correspond to rotations of $w$ that start with a symbol $c\in\Sigma$.
    Then,
    $x[i..j]$ will consist of
    $|w|_c - \rho_c(w)$ occurrences of $c$'s that correspond
    to those preceding a $c$ in the same run of $c$'s in $w$,
    and $\rho_c(w)$ occurrence of symbols not equal to $c$ that precede a maximal run of $c$'s in $w$.
    Thus, $\rho(x[i..j])$ is maximized
    when the $\rho_c(w)$ non-$c$ symbols in $x[i..j]$ are not adjacent to each other.
    Mantaci et al. argued that
    this implies
    $\rho(x[i..j]) \leq 2\rho_c(w)$
    summing up to $2\rho(w)$ for all $c$.
    However, they did not give a reason to why $\rho(x[i..j])$ could not be $2\rho_c(w)+1$ (summing up to $\rho(w)+\sigma$), which is possibly the case when $x[i..j]$ starts and ends with the symbol $c$.
    We show that in such a case, an
    LF mapping with a single cycle cannot be defined, and ultimately,
    $\rho(x(w)[i..j]) \leq 2\rho_c(w)$ holds.

    Let $w[i'..j']$ be a maximal run of symbol $c$ in $w$ and let
    $k$ be such that $x[k]$ corresponds to $w[j']$, i.e.,
    $k$ is the $\omega$-order rank of $\rot^{j'+1}(w)$.
    Since $w[j'+1]\neq c$, $k \not\in [i'..j']$.
    Also, since $w[i]=\cdots=w[j]=c$,
    the $\omega$-order ranks $k, \Psi_x(k), \ldots, \Psi^{j'-i'+1}_x(k)$,
    corresponding respectively to those of $\rot^{j'+1}(w), \ldots, \rot^{i'}(w)$,
    are in $[i..j]$ (except $k$)
    and monotonic, i.e., either
    $k < i\leq \Psi_x(k) < \cdots < \Psi^{j'-i'+1}_x(k) \leq j$ or
    $k > j \geq \Psi_x(k) > \cdots > \Psi^{j'-i'+1}_x(k)\geq i$ holds.
    Also, $x[\Psi^{j'-i'+1}_x(k)] = w[i'-1] \neq c$.
    On the other hand, for any $p,q$ s.t. $i\leq p<q \leq j$ and $x[p]=x[q]=c$,
    $\Psi_x(p) < \Psi_x(q)$ must hold
    implying that for any such $p,q$, it cannot be that
    $p$ is part of a monotonically decreasing cycle and
    $q$ is part of a monotonically increasing cycle.
    Since each of the above monotonic sequences end at a non-$c$ position in $x[i..j]$,
    there must be some position $k'\in [i..j]$ such that $\Psi_x(k') = k'$,
    and $\Psi_x$ cannot be defined to cover all positions in $[i..j]$ in a single cycle.
    Therefore $\rho(x[i..j]) \leq 2\rho_c(w)$ holds.

    For $x = \bbwt(w)$, similarly consider the range $[i..j]$ of $\bbwt(w)$ such that
    the corresponding rotations of the Lyndon factors of $w$
    start with the symbol $c\in\Sigma$.
    Since it is known that Lyndon factors of $w$ (except for single symbol factors) must start and end at run-boundaries of $w$~\cite{DBLP:conf/isaac/FujishigeNIBT17},
    a maximal run of $c$'s in $w$ is
    either a substring of a Lyndon factor of $w$, or a maximal run of a single symbol Lyndon factors.
    If it is a substring of a Lyndon factor but not a prefix, then there is one non-$c$ symbol per such run that precedes the run and occurs in $x[i..j]$.
    If it is a prefix of a Lyndon factor longer than $1$, then
    the previous symbol (or the last symbol of the Lyndon factor) is non-$c$
    due to the Lyndon factor not having a proper border.
    Otherwise, the maximal run is a maximal run of single symbol Lyndon factors $c$,
    in which case the previous symbol for these occurring in $x[i..j]$ will also be $c$.
    Thus, the arguments in the previous paragraph for BWT hold for BBWT as well, except for the last part:
    it can be that
    $x[i]=x[j]=c$ and
    $\Psi(k')=k'$ for some $k'\in[i,j]$.
    However, this implies that this
    corresponds to a single symbol
    Lyndon factor which
    would not have introduced a non-$c$ symbol in $x[i..j]$.
    Therefore
    $\rho(x[i..j])\leq 2\rho_c(w)$ still holds.
\end{proof}

\subsection{In terms of other repetitiveness measures}
We consider measuring the clustering effect of $\bwt$ and $\bbwt$ other than~$\rho$.
Here, we show that $\bwt$ or $\bbwt$
can only increase the repetitiveness of the string by a polylogarithmic factor.

\begin{theorem}\label{thm:clusteringEffectByOthers}
    For repetitiveness measure $\mu \in \{\delta,\gamma,b,v,z, g_{\mathrm{rl}}\}$
    and any string $w$ of length $n$,
    it holds that
    $\mu(\bwt(w)) = O(\mu(w)\log^2 n)$.
    Moreover,
    $\mu(\bbwt(w)) = O(\mu(w)\log^2 n)$
    for $\mu \in \{z, g_{\mathrm{rl}}\}$,
    and
    $\mu(\bbwt(w)) = O(\mu(w)\log^6 n)$
    for $\mu \in \{\delta,\gamma,b,v\}$.
\end{theorem}
\begin{proof}
    The proof for $\bwt$ follows from a simple observation that
    for any string $w$, $\mu(w) = O(\rho(w))$
    implying  $\mu(\bwt(w)) = O(\rho(\bwt(w))$.
    Since $r(w) = \rho(\bwt(w)) = O(\delta(w)\log^2n)$,
    we have
    $\mu(\bwt(w)) = O(\delta(w)\log^2 n) = O(\mu(w)\log^2 n)$,
    where the last relation
    follows
    from the fact that $\delta$ lower bounds all the repetitiveness measures considered.

    For $\bbwt$,
    $\mu(\bbwt(w)) = O(\rho(\bbwt(w)))$ holds, but only
    $r_B(w) = \rho(\bbwt(w)) = O(z(w)\log^2 n)$ has been proved,
    from which
    $r_B(w) = O(g_{\mathrm{rl}}(w)\log^2 n)$ and the statement follows for $\mu\in\{z,g_{\mathrm{rl}}\}$.
    For the other measures,
    \Cref{corollary:RvsRB} gives us
    $r_B(w) = O(r(w)\log^4 n) =
        O(\delta(w)\log^6 n)$
    from which the statement follows.
\end{proof}

\subsection{A family of significantly ``clustered'' strings by BBWT}
Finally, we show an infinite family of strings for which
BBWT exhibits an asymptotically maximal clustering effect in terms of $\delta, \gamma, b, v, r$, and slightly weaker in terms of $z,g_{\mathrm{rl}},g$.
Namely,
we consider the family of strings whose
BBWT image are Fibonacci words.
\begin{restatable}{theorem}{bbwtTwice}\label{theorem:bbwt_twice}
    Let $w$ be a string where
    $\bbwt(w)$ is a Fibonacci word of length $n$.
    Then, for $\mu\in\{\delta,\gamma,b,v,r\}$,
    $\mu(\bbwt(w)) = O(\frac{\log n}{n}\mu(w))$
    and for
    $\mu\in\{z, g_{\mathrm{rl}},g \}$,
    $\mu(\bbwt(w)) = O(\frac{\log^2 n}{n}\mu(w))$.
\end{restatable}

In order to understand $w$,
we first introduce the tools we use to
characterize the LF mapping $\Psi_{F_k}$ on Fibonacci words.

The Zeckendorf representation~\cite{1570009749187075840} of a non-negative integer is a
unique (sub-)set of distinct non-consecutive Fibonacci numbers $\{ f_k \mid k \geq 1 \}$ that sum up to the integer.
We represent a non-negative integer $i$ as a bit string $Z(i)$,
where $i = ||Z(i)|| = \sum_{j\geq 0} Z(i)[j]\cdot f_{j+1}$.
We use $Z_k(i)$ to denote the length-$k$ prefix $Z(i)[0..k-1]$ of $Z(i)$.
Note that for any $i \in [0,f_k)$,
it suffices that a subset of
$\{ f_1,\dots,f_{k-1}\}$ is used, and thus $Z(i)$ requires only up to $(k-1)$ bits,
i.e., $1$'s will only occur in $Z(i)[0..k-2]$, and the rest $Z(i)[k-1..]$ will be $0$.

We observe that length $i\geq 0$ prefix of a Fibonacci word can be uniquely factorized into a sequence of distinct non-consecutive Fibonacci words in order of decreasing length, where the length of each Fibonacci word corresponds to an element in the Zeckendorf representation of~$i$.
\begin{lemma}\label{lemma:fibonacciStringWithNumeration}
    For any $i \geq 0$,
    \begin{align}
        F_\infty[0..i) & = \prod_{j={k-2}}^0 F_{j+1}^{Z(i)[j]} = F_{k-1}^{Z(i)[k-2]}\cdots F_{1}^{Z(i)[0]},
    \end{align}
    where $k$ is the smallest integer such that $i < f_k$.
\end{lemma}
\begin{proof}
    Consider a greedy factorization of $F_\infty[0..i)$ that
    takes the longest prefix of the remaining string that is a Fibonacci word.
    Let $k$ be the smallest integer such that $i \leq f_k$.
    If $i = f_k$ then we simply take $F_k$ as the last element.
    Otherwise, we take $F_{k-1}$.
    Notice that in this case, the remaining string is
    $F_\infty[f_{k-1}..i) = F_k[f_{k-1}..i) = F_{k-2}[0..i-f_{k-1})$ of length $i'=i - f_{k-1} < f_{k-2}$,
    and we repeat the process to find a greedy factorization of $F_{k-2}[0..i')$.
    The remaining string is a proper prefix of $F_{k-2}$
    and thus $F_{k-2}$ will not be chosen next, implying that consecutive
    Fibonacci words are not chosen.
    Thus, the lengths of the sequence of strings will be a set of non-consecutive set of distinct Fibonacci numbers that sum up to $i$.
    The lemma follows from the uniqueness of the Zeckendorf representation of~$i$.
\end{proof}

The following relation about the least significant bit in the Zeckendorf representation and the $i$th symbol in the
(infinite) Fibonacci word is known.
\begin{lemma}[Problem 6 in~\cite{Crochemore_Lecroq_Rytter_2021}]\label{lemma:firstBitOfZ}
    For $i = 0,1,\dots$,
    \begin{align*}
        F_\infty[i] & = \begin{cases}
                            \mathtt{a} & \mbox{if } Z(i)[0] = 0, \\
                            \mathtt{b} & \mbox{if } Z(i)[0] = 1.
                        \end{cases}
    \end{align*}
\end{lemma}

Now, we make observations on the LF-mapping for Fibonacci words.

\begin{lemma}\label{lemma:PsiOfFibonacci}
    The LF-mapping function for the Fibonacci word $F_k$ is, for any $0 \leq i < f_k$,
    \[
        \Psi_{F_k}(i) = \begin{cases}
            \rank_\mathtt{a}(i,F_k)           & \mbox{if } F_k[i] = \mathtt{a}, \\
            f_{k-1} + \rank_\mathtt{b}(i,F_k) & \mbox{if } F_k[i] = \mathtt{b}.
        \end{cases}
    \]
\end{lemma}
\begin{proof}
    Straightforward from the definition of $\Psi_{F_k}$ and \Cref{observation:NumberOfAorB}.
\end{proof}

The next lemma shows that the LF mappings on $F_k$ can be interpreted
as rotation operations on the Zeckendorf representation of the position to be mapped.
\begin{lemma}\label{lemma:PsiOfFibIsRotation}
    For any $i \in [0,f_k)$,
    \begin{align*}
        Z_k(\Psi_{F_k}(i)) & = \begin{cases}
                                   \rot(Z_k(i))   & \text{if $F_k[i] = \mathtt{a}$}, \\
                                   \rot^2(Z_k(i)) & \text{if $F_k[i] = \mathtt{b}$.}
                               \end{cases}
    \end{align*}
\end{lemma}
\begin{proof}
Let $j = \Psi_{F_k}(i)$.
    Since $i, j\in [0,f_{k})$, the most significant bits
    of $Z_k(i)$ and $Z_k(j)$, corresponding to $f_{k}$, are always zero, i.e.,
    $Z_k(i)[k-1]= Z_k(j)[k-1] = 0$.

    From Lemma~\ref{lemma:firstBitOfZ}, $F_k[i] = \mathtt{a}$ implies that $Z_k(i)[0] = 0$.
    Thus, $\rot(Z_k(i)) = Z_k(i)[1..k-1]\cdot 0$.
Therefore, we have
    \begin{align*}
        \Psi_{F_k}(i) & =\rank_\mathtt{a}(i,F_k)                                                &  & \text{by \Cref{lemma:PsiOfFibonacci}}                \\
                      & =\left|\prod_{j=k-2}^0F_{j+1}^{Z(i)[j]}\right|_\mathtt{a}               &  & \text{by \Cref{lemma:fibonacciStringWithNumeration}} \\
                      & =\sum_{j=0}^{k-2}Z_k(i)[j]\cdot f_j =\sum_{j=0}^{k-1}Z_k(i)[j]\cdot f_j                                                           \\
                      & =Z_k(i)[0]\cdot f_k + \sum_{j=0}^{k-2}Z_k(i)[j+1]\cdot f_{j+1}                                                                    \\
                      & =\sum_{j=0}^{k-1}Z_k(i)[(j+1)\bmod k]\cdot f_{j+1}                                                                                \\
                      & =||\rot(Z_k(i))||
    \end{align*}
    On the other hand, again from \Cref{lemma:firstBitOfZ},
    $F_k[i] = \mathtt{b}$ implies that $Z_k(i)[0] = 1$
    and since the Zeckendorf representation does not use consecutive Fibonacci numbers,
    $Z_k(i)[1]=0$.
    Thus, $\rot^2(Z_k(i)) = Z_k(i)[2..k-1]\cdot 1\cdot0$.
    Therefore, we have
\begin{align*}
        \rank_\mathtt{b}(i,F_k)
         & =\left|\prod_{j=k-2}^0F_{j+1}^{Z(i)[j]}\right|_\mathtt{b}                      &  & \text{by \Cref{lemma:fibonacciStringWithNumeration}} \\
         & =\sum_{j=0}^{k-2}Z_k(i)[j]\cdot f_{j-1}=\sum_{j=0}^{k-1}Z_k(i)[j]\cdot f_{j-1}                                                           \\
         & =Z_k(i)[0]\cdot f_{-1} + Z_k(i)[1]\cdot f_{0} +
        \sum_{j=0}^{k-3}Z_k(i)[j+2]\cdot f_{j+1}                                                                                                    \\
         & = \sum_{j=0}^{k-3}Z_k(i)[j+2]\cdot f_{j+1}
    \end{align*}
    Furthermore,
    \begin{align*}
        \Psi_{F_k}(i) & =\rank_\mathtt{b}(i,F_k) + f_{k-1}                                                         &  & \text{by~\Cref{lemma:PsiOfFibonacci}} \\
                      & = \sum_{j=0}^{k-3}Z_k(i)[j+2]\cdot f_{j+1} + f_{k-1}                                                                                  \\
                      & = \sum_{j=0}^{k-3}Z_k(i)[j+2]\cdot f_{j+1} + Z_k(i)[0]\cdot f_{k-1} + Z_k(i)[1]\cdot f_{k}                                            \\
                      & = \sum_{j=0}^{k-1}Z_k(i)[(j+2)\bmod k]\cdot f_{j+1}                                                                                   \\
                      & = ||\rot^2(Z_k(i))||.
    \end{align*}
    thus proving the lemma (see also~\Cref{fig:Fib-LFmap}).
\end{proof}

\begin{figure}[t]
    \begin{center}
        \includegraphics[width=0.8\textwidth]{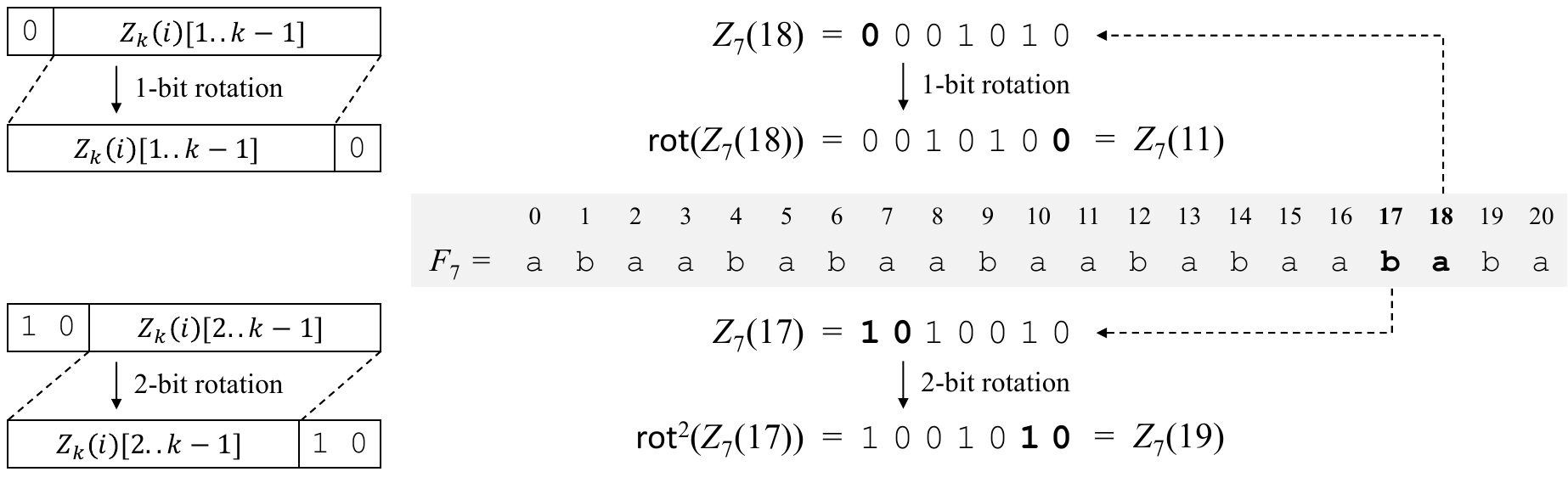}
        \caption{
        Example of Lemma~\ref{lemma:PsiOfFibIsRotation} for $F_7$.
        The $\mathtt{a}$ at position $18$ is the $12$th $\mathtt{a}$ in $F_7[0..18]$,
        and thus the LF mapping should point to position $11$, whose Zeckendorf representation $Z_7(11)$ is a
        1-bit left rotation of $Z_7(18)$.
        The $\mathtt{b}$ at position $17$ is the $7$th $\mathtt{b}$ in $F_7[0..17]$,
        and since there are $13$ $\mathtt{a}$'s in $F_7$, the LF mapping should point to position $13+7-1=19$,
        whose Zeckendorf representation $Z_7(19)$ is a 2-bit left rotation of $Z_7(17)$.
        }
        \label{fig:Fib-LFmap}
    \end{center}
\end{figure}

We are now ready to prove~\Cref{theorem:bbwt_twice}:
\bbwtTwice*
\begin{proof}
    For any Fibonacci word $F_k$,
    $\delta(F_k) \leq \gamma(F_k) \leq b(F_k) \leq v(F_k) \leq r(F_k) = 2$
    and $z(F_k),g_{\mathrm{rl}}(F_k),g(F_k) = O(\log n)$ are known.
    In the rest of the proof, we show
    $\delta(w) = \Omega(n/\log n)$ implying
    the same lower bound for all other measures
    which finishes the proof.

    From Lemma~\ref{lemma:PsiOfFibIsRotation}, $\Psi_{F_k}$ can be interpreted as a rotation operation on a $k$-bit string corresponding to the Zeckendorf representation of the given position.
    Therefore, the number of cycles that $\Psi_{F_k}$ produces is equivalent to the number of conjugacy classes among
    bit strings corresponding to the Zeckendorf representations of the set of positions $[0,F_k)$.
    Since the Zeckendorf representations do not use adjacent Fibonacci numbers,
    this is equivalent to the number $C(k)$ of binary necklaces of length $k$ that do not contain ``11''.
    This corresponds to entry A000358 in the on-line encyclopedia of integer sequences,
    which is known~\cite{bona2015handbook} to be:
    \begin{align}\label{eqn:ck}
        C(k) & = \frac{1}{k}\sum_{d|k}\varphi(k/d)(f_{d-2}+f_{d}),
    \end{align}
    where
    $d|k$ represents that $d$ is a divisor of $k$,
    and $\varphi$
    is Euler's totient function which returns,
    for integer $n$, the number of positive integers less than $n$ that are relatively prime to $n$ (i.e., having $\gcd$ of $1$).

    As seen in the proof of Lemma~\ref{lemma:PsiOfFibIsRotation},
    an $\mathtt{a}$ in $F_k$ implies
    a 1-bit left rotation on $Z_k$ for $\Psi_{F_k}$,
    and $\mathtt{b}$ in $F_k$ implies
    a 2-bit left rotation on $Z_k$ for $\Psi_{F_k}$.
    For simplicity of the analysis, we will restrict $k$ to being prime.
    Each necklace will then correspond to some cyclic string $w'$, where
    \begin{align}
        |w'|_\mathtt{a} +2|w'|_\mathtt{b} & = k, \label{eqn:absumsto:k}
    \end{align}
    except for the string $\mathtt{a}$ corresponding to the bit string $0^k$,
    since all other necklaces must be primitive and a total of $k$-bits must be rotated in order to come back to the same bit string.
    Consider the lexicographically smallest rotations of all of them, which are Lyndon words.
    All of them are distinct,
    and any two of them (except for the single $\mathtt{a}$)
    cannot be a substring of the other due to the constraint on their lengths (\Cref{eqn:absumsto:k}).
    The string
    $w = \bbwt^{-1}(F_k)$ is a concatenation (in non-increasing order) of all
    these strings with the single $\mathtt{a}$ appended at the end,
    corresponding to the Lyndon factorization of $w$.
    Since Lyndon words
    can only occur as a substring of a Lyndon factor,
    it follows that
    there are
    $\Omega(C(k))$ distinct uniquely occurring strings (Lyndon factors) that do not overlap.
For prime $k$, \Cref{eqn:ck} becomes:
    \begin{align*}
        C(k) & = \frac{1}{k} \left(\varphi(1)(f_{k-2}+f_{k}) + \varphi(k)(f_{-1}+f_1)\right)
        = \frac{1}{k} (f_{k-2}+f_{k} +k-1) \geq f_k/k.
    \end{align*}

    Since the uniquely occurring Lyndon factors of $w$ satisfy
    \Cref{eqn:absumsto:k} and have
    length at most
    $k-1$, any length $2k$ substring of $w$ will contain at least one full occurrence of a Lyndon factor of $w$.
    Since all such strings must also have unique occurrences and therefore be distinct,
    $\delta(w) \geq (f_k-2k)/2k = \Omega(n/\log n)$ holds.
\end{proof}

For example, for $\bbwt^{-1}(F_7)$,
the binary necklaces of length $7$ that do not contain ``11'' are:
$0000000$, $0000001$, $0000101$, $0001001$, $0010101$,
which respectively correspond to
$\mathtt{a}$, $\mathtt{aaaaab}$, $\mathtt{aaabb}$,
$\mathtt{aabab}$, and $\mathtt{abbb}$, and therefore,
$\bbwt^{-1}(F_7) = \mathtt{abbbaababaaabbaaaaaba}$.

We note that $\bbwt^{-1}(F_k)$ is
asymptotically maximally incompressible
by dictionary compression since
it is known that for any string,
$\delta \leq \gamma \leq z = O(n/\log_\sigma n)$ holds,
which is $O(n/\log n)$ for binary strings.

\section{Discussion}

The $\log$ factors in our bounds, especially in \Cref{thm:clusteringEffectByOthers}
are most likely not tight;
Our focus here was on showing the existence of such bounds.
A natural open question is how many $\log$ factors can be shaved.

\clearpage
\bibliography{refs}

\end{document}